\documentclass[twoside,leqno,twocolumn]{article}

\usepackage[letterpaper]{geometry}

\usepackage{siamproceedings}

\usepackage[T1]{fontenc}
\usepackage{amsfonts}
\usepackage{graphicx}
\usepackage{epstopdf}
\usepackage{enumitem}
\usepackage{algorithmic}
\ifpdf
  \DeclareGraphicsExtensions{.eps,.pdf,.png,.jpg}
\else
  \DeclareGraphicsExtensions{.eps}
\fi

\newsiamremark{remark}{Remark}
\newsiamremark{hypothesis}{Hypothesis}
\Crefname{hypothesis}{Hypothesis}{Hypotheses}
\newsiamthm{claim}{Claim}

\usepackage{amsopn}

\usepackage{xspace}
\usepackage{float}
\usepackage{cite}
\usepackage{subcaption}
\usepackage{booktabs}
\usepackage[table]{xcolor}

\newcommand{\xth}[1]{$#1^{\text{th}}$}
\newcommand{\ceil}[1]{\left\lceil #1 \right\rceil}

\newcommand{\word}{\mathcal{W}}
\newcommand{\lcp}{\text{lcp}\xspace}

\newcommand{\BWT}{\ensuremath{\mathrm{BWT}}\xspace}
\newcommand{\RLBWT}{\ensuremath{\mathrm{RLBWT}}\xspace}
\newcommand{\RLF}{\ensuremath{\mathrm{RF}}\xspace}
\newcommand{\SA}{\ensuremath{\mathrm{SA}}\xspace}
\newcommand{\ISA}{\ensuremath{\mathrm{ISA}}\xspace}

\newcommand{\LF}{\ensuremath{\mathrm{LF}}\xspace}
\newcommand{\FL}{\ensuremath{\mathrm{FL}}\xspace}
\newcommand{\LCP}{\ensuremath{\mathrm{LCP}}\xspace}
\newcommand{\ILCP}{\ensuremath{\mathrm{LCP}^{+}}\xspace}
\newcommand{\PLCP}{\ensuremath{\mathrm{PLCP}}\xspace}
\newcommand{\IPLCP}{\ensuremath{\mathrm{PLCP}^{+}}\xspace}

\DeclareMathOperator{\pred}{pred}
\DeclareMathOperator{\successor}{succ} 
\DeclareMathOperator{\rank}{rank}

\DeclareMathOperator{\outputint}{\mathcal{O}}

\newsiamthm{invariant}{Invariant}
\crefname{invariant}{Invariant}{Invariants}
\Crefname{invariant}{Invariant}{Invariants}
\crefalias{invariant}{invariant}

\newcommand{\PRANK}{P_{\text{r}}}
\newcommand{\QRANK}{Q_{\text{r}}}
\newcommand{\MOVE}[1]{\mathcal{M}_{#1}}

\newcommand{\PLIST}{\mathcal{L}_P}
\newcommand{\QLIST}{\mathcal{L}_Q}

\begin{document}

\title{\Large Optimal-Time Move Structure Construction}
   \author{Nathaniel K. Brown\thanks{Department of Computer Science, Johns Hopkins University, Baltimore, MD, USA (\email{nbrown99@jh.edu}, \email{blangme2@jhu.edu}). N. K. Brown is supported by NSERC PGS-D and NIH grant R01HG011392; B. Langmead is supported by NIH grant R01HG011392.}
    \and Ahsan Sanaullah\thanks{Department of Computer Science, University of Central Florida, Orlando, FL, USA (\email{ahsan.sanaullah@ucf.edu}, \email{shzhang@cs.ucf.edu}). The majority of A. Sanaullah's work was completed while affiliated with the University of Central Florida; their current affiliation is Department of Data Science, Dana-Farber Cancer Institute, Boston, MA, USA. Both authors are supported by NIH grant R01HG010086.}
    \and Shaojie Zhang\footnotemark[2]
    \and Ben Langmead\footnotemark[1]}

\date{}








\maketitle

\begin{abstract}
The move structure represents a permutation $\pi$ of $[0,n)$ by partitioning the domain into $O(r)$ disjoint, contiguously permuted intervals, with $r$ being the minimum number of such intervals. This data structure occupies $O(r)$ words of space and enables $O(1)$-time computation of $\pi(i)$ given the interval that contains $i$. For permutations where $r \ll n$, this provides an efficient, compressed representation for navigation. While the existing best $O(r)$-space construction approaches require $O(r\log r)$-time, we present an optimal $O(r)$-time and space construction algorithm. This is achieved by replacing balanced search trees with pointer-based lists through introducing a bidirectional strategy that synchronizes construction of the structures for $\pi$ and its inverse $\pi^{-1}$ in a single, unified pass. By applying this algorithm to an existing result, we achieve the first optimal $O(n)$-time construction of the longest common prefix (LCP) array from a run-length-encoded Burrows-Wheeler transform (RLBWT) of $r$ runs in $O(r)$ working space. Empirical evaluation on pangenome-scale data confirms that our move structure construction algorithm is consistently faster than the previous best, achieving speedups of up to $\sim 2\times$ while using less memory.
\end{abstract}




\section{Introduction}
Nishimoto and Tabei's \textit{move structure}~\cite{nishimoto2021optimal} represents run-length compressible permutations in compact space. Formally, for a permutation $\pi$ over $[0,n)$, let $r = 1 + |\{i\in [1,n) : \pi(i) - 1 \neq \pi(i-1)\}|$ be the size of the minimal partition of $[0,n)$ into disjoint, contiguously permuted intervals. This structure requires $O(r)$ words of space and enables $O(1)$-time evaluation of $\pi(i)$ given the interval containing $i$. This constant query time is achieved via a ``balancing'' phase that splits some intervals, with the best known construction algorithm taking $O(r\log r)$-time and $O(r)$-space given a run-length encoded input~\cite{nishimoto2021optimal}.
The move structure has been used extensively in recent algorithmic developments using \textit{Burrows-Wheeler transform} (\BWT) based compressed text indexing strategies~\cite{nishimoto2021optimal, brown2022rlbwt, bertram2024move, depuydt2025b, brown2026col, kimura_et_al:LIPIcs.CPM.2026.10, zakeri2026movi}.

The \BWT~\cite{bwt} is fundamental to efficient text processing. When computed over a repetitive text collection of total length $n$, it often contains relatively fewer runs $r$, leading to the \textit{run-length encoded BWT} (\RLBWT). Consequently, researchers have developed techniques operating in $O(r)$-space~\cite{gagie2020fully,nishimoto2021optimal}. In genomics, where $r$ often grows sublinearly with respect to $n$~\cite{kuhnle2019efficient}, these structures are essential for managing pangenome-scale data where the size of text collections can reach the order of trillions~\cite{ropebwt3,zakeri2026movi,zakeri2024movi,sanaullah2026rlbwt}.

Indexes built on repetitive texts rely on the $\RLBWT$-based permutations $\LF, \FL, \phi$ and $\phi^{-1}$ to perform efficient pattern-matching queries. Where $r$ is the number of runs in the \RLBWT, they can be represented as move structures in $O(r)$-space with $O(1)$ step query time~\cite{nishimoto2021optimal}. Sanaullah et al.~\cite{sanaullah2026rlbwt} leveraged this to compute the longest common prefix array (\LCP) from an \RLBWT in $O(n+r \log r)$-time and $O(r)$ working space. Their complexity, however, is constrained by the $O(r \log r)$ construction time of the underlying move structures.

In this paper, we present an optimal $O(r)$-time and space construction algorithm for the move structure. This is achieved by replacing balanced search trees with pointer-based lists via a synchronized construction of the structures for $\pi$ and $\pi^{-1}$ in a single, unified pass. Consequently, applying this to Sanaullah et al.'s~\cite{sanaullah2026rlbwt} approach provides  the first optimal $O(n)$-time algorithm to compute the LCP array from the RLBWT using $O(r)$ working space. Finally, empirical evaluation demonstrates that our construction is consistently faster, achieving speedups of up to $\sim 2\times$ over the previous best while using less memory.

\section{Preliminaries}
\subsection{Notation}

A \textit{predecessor query} on a sorted array $A$ is defined as $A.\pred(x)=\max \{y\in A~|~y\leq x\}$ and a \textit{successor query} as $A.\successor(x)=\min \{y\in A~|~y > x\}$. We treat sorted arrays of distinct elements interchangeably as 0-indexed arrays and as ordered sets, permitting standard set-theoretic operations directly over array elements. For a sorted array $A$ of distinct elements, the rank query $j=A.\rank(x)$ denotes the 0-based index of $y=A.\pred(x)$ in sorted order. A \textit{string} $S[0..n-1]$ is an array of symbols from an ordered alphabet $\Sigma$ with $|\Sigma| = \sigma$. Let $S[i..j]$ represent the \textit{substring} $S[i] \dots S[j]$. A \textit{prefix} of $S$ is a substring $S[0..j]$, whereas a \textit{suffix} is a substring $S[i..n-1]$. Let $\lcp(S_1,S_2)$ denote the length of the longest common prefix between strings $S_1$ and $S_2$. A \textit{text} $T[0..n-1]$ is assumed to be terminated by the special symbol $\$ \notin \Sigma$ of smallest order so that suffix comparisons are well defined. We use the word RAM model, assuming machine words of size $\word=\Theta(\log n)$ with basic arithmetic and logical bit operations in $O(1)$-time. Space complexities in this paper are in words unless the number of bits is specified.

\subsection{Suffix Array and Burrows-Wheeler Transform}

The \textit{suffix array} (\SA) of a text $T[0..n-1]$ is  $\SA[0..n-1]$ such that $\SA[i]$ is the starting position of the \xth{i} lexicographically smallest suffix of $T$. The \textit{inverse suffix array} (ISA) is its inverse such that $\ISA[i]$ is the lexicographic rank of the suffix starting at $i$ in $T$. The \textit{Burrows-Wheeler Transform} (\BWT)~\cite{bwt} is a permutation $\BWT[0..n-1]$ of a text $T[0..n-1]$ such that $\BWT[i] = T[(\SA[i] - 1) \mod n]$. It is reversible by the \textit{last-to-first} (\LF) mapping $\LF(i)$, a permutation over $[0,n)$ satisfying $\SA[\LF(i)]=(\SA[i]-1)\mod n$. The \textit{first-to-last} (\FL) mapping is its inverse satisfying $\SA[\FL(i)]=(\SA[i]+1)\mod n$. Let $r$ be the number of maximal equal character runs of the \BWT. The \textit{run-length encoded BWT} (\RLBWT) is an array $\RLBWT[0..r-1]$ of tuples where $\RLBWT[i].c$ is the character of the \xth{i} \BWT run and $\RLBWT[i].\ell$ its length.

The permutation $\phi$ over $[0,n)$ is defined such that $\phi(\SA[i])=\SA[(i-1) \mod n]$, returning the \SA value for the suffix of preceding lexicographic rank. Its inverse $\phi^{-1}$ is defined symmetrically such that $\phi^{-1}(\SA[i])=\SA[(i+1) \mod n]$. The \textit{longest common prefix array} (\LCP) stores the \lcp between lexicographically adjacent suffixes in $T[0..n-1]$, i.e., $\LCP[0..n-1]$ such that $\LCP[0]=0$ and $\LCP[i]=\lcp(T[\SA[i-1]..n-1], T[\SA[i]..n-1])$ for $i>0$. The \textit{permuted LCP array} (\PLCP) reorders the \LCP array in text order such that $\PLCP[i]=\LCP[\ISA[i]]$. An \textit{irreducible LCP} is an $\LCP[i]$ where $i$ is the start of a \BWT run, and an \textit{irreducible PLCP} is its equivalent $\PLCP[\SA[i]]$. For $r$ \BWT runs, we refer to the arrays of irreducible LCP and PLCP values as $\ILCP[0..r-1]$ and $\IPLCP[0..r-1]$ respectively.

\subsection{LCP Array from RLBWT}
\label{sec:lcp_rlbwt}

Kasai et al.~\cite{kasai2001linear} gave the first optimal $O(n)$-time algorithm to construct the \LCP array from the \SA. Subsequently, K\"arkk\"ainen, Manzini, and Puglisi~\cite{karkkainen2009permuted} introduced the $\phi$ algorithm (named for its use of the $\phi$ permutation) to compute the $\PLCP$ array as an intermediate step to recover the $\LCP$ array. However, both approaches require $O(n)$-space for random access to the text. To overcome this, Beller et al.~\cite{beller2013computing} and its extension by Prezza and Rosone~\cite{prezza2019space} avoid random access to the text by using a wavelet tree to compute the \LCP array from the \BWT in $O(n \log \sigma)$-time. This algorithm operates in succinct space: $o(n \log \sigma)$ bits of working space in addition to the input/output.

Recently, Sanaullah et al.~\cite{sanaullah2026rlbwt} showed how to compute the \LCP array from the \RLBWT in $O(n + r \log r)$-time and $O(r)$ working space by adapting the $\phi$ algorithm to use move structures. Their method relies on the observation that only the $O(r)$ irreducible $\PLCP$ values are needed; when sampled alongside a $\phi$/$\phi^{-1}$ move structure, the remaining \PLCP values can be efficiently recovered. Further, they simulate random access to the text, as required in the $\phi$ algorithm, by using a move structure for $\FL$. Their key insight is that by sampling $\ISA$ values, they can navigate the $\FL$ move structure to compute irreducible $\PLCP$ values in $O(n)$-time using $O(r)$-space. 

However, the overall time complexity of this approach inherits an $O(r \log r)$ term from constructing a balanced move structure to support \FL queries, resulting in $O(n + r \log r)$-time and $O(r)$ working space for \LCP array computation. Although length capping~\cite{brown2026boundingaveragestructurequery} enables optimal-time performance for operations like suffix array enumeration and BWT inversion by using an alternative move structure construction, its guarantees alone do not establish the same bound for the $\phi$ algorithm. Thus, an $O(r)$-time and space balanced move structure construction would suffice to obtain optimal-time \LCP construction using their method. We summarize the necessary foundational results from Sanaullah et al.~\cite{sanaullah2026rlbwt} below; complete proofs and algorithmic details are deferred to the Appendix.

\begin{theorem}[Sanaullah et al.~\cite{sanaullah2026rlbwt}]\label{thm:irreducible_og}
    Given $\RLBWT[0..r-1]$ over $T[0..n-1]$, the irreducible \PLCP array $\IPLCP[0..r-1]$ can be constructed in $O(n + r \log r)$-time and $O(r)$-space.
\end{theorem}

\begin{theorem}[Sanaullah et al.~\cite{sanaullah2026rlbwt}]\label{thm:lcp_og}
    Given $\RLBWT[0..r-1]$ over $T[0..n-1]$, $\LCP[0..n-1]$ can be constructed in $O(n + r\log r)$-time and $O(r)$ working space.
\end{theorem}

\subsection{Move Structure}
\label{sec:movestructure}
Given a permutation $\pi$ over $[0,n)$, let the set of \textit{input intervals} $I=\{I[0],\dots,I[r-1]\}$ be the partition of $[0,n)$ into $r$ maximal disjoint, contiguously permuted intervals ordered such that $\min(I[j])<\min(I[j+1])$. Let $P[0..r]$ be the sorted array built over the interval starting positions $\{0\} \cup \{i \in[1,n)~|~\pi(i-1) \neq \pi(i)-1\}$ appended with the boundary element $P[r]=n$. Notice that $P$ implicitly represents $I$ such that $I[j]=[P[j],P[j+1])$ and $P[j]=\min(I[j])$. These intervals identify redundancy, since they permute contiguously: if $i \in I[j]$ then $\pi(i)=\pi(P[j]) + (i-P[j])$. For any position $i\in[0,n)$, if $j=P.\rank(i)$ then $i\in I[j]$. Hence, computing $\pi(i)$ from only $P$ inherits predecessor query bounds, yielding non-constant query time~\cite{belazzougui2015optimal}.

Given $i$ and $j=P.\rank(i)$, Nishimoto and Tabei's \textit{move structure}~\cite{nishimoto2021optimal} computes $\pi(i)$ and $P.\rank(\pi(i))$ in $O(r)$-space and $O(1)$-time. Let $P_\pi[0..r-1]$ be defined such that $P_\pi[j]=\pi(P[j])$. Knowing $i\in I[j]$, we can easily compute $\pi(i)=P_\pi[j]+(i-P[j])$. To return the input interval containing $\pi(i)$ without a predecessor query, they store $\PRANK[0..r-1]$ where $\PRANK[j]=P.\rank(P_\pi[j])$; intuitively, $\PRANK[j]$ is the index of the input interval containing $\pi(P[j])$. Since $\pi(i) \geq \pi(P[j])$, we perform a linear scan of $P$ starting from $\PRANK[j]$ to find the largest $j'$ such that $P[j'] \leq \pi(i)$, thereby identifying the interval containing $\pi(i)$. Although $P$, $P_\pi$, and $\PRANK$ are $O(r)$-space, bounding the linear scan requires a more careful construction. Overall, this operation is defined as a \textit{move query}; see Figure~\ref{fig:move_structure} for a visualization of the underlying array structures and interval decomposition.

\begin{figure}[htbp]
    \centering
    \begin{minipage}[b]{\columnwidth}
        \centering
        \includegraphics[width=\linewidth]{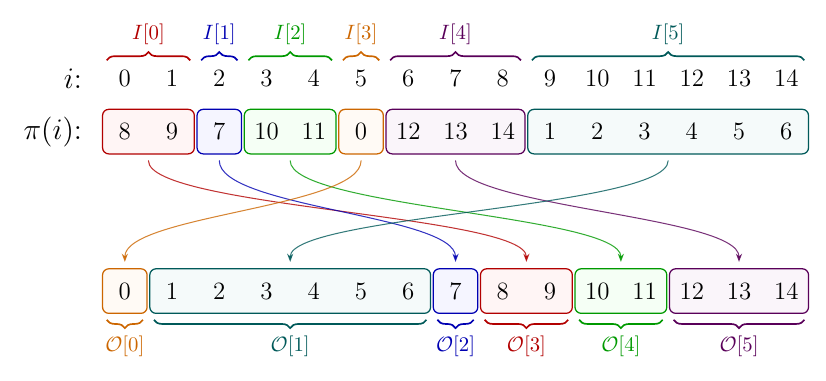}
        \subcaption{Visual representation of move structure intervals.}
    \end{minipage}%
    \hspace{0.03\textwidth}
    \begin{minipage}[b]{\columnwidth} 
        \centering
        \begingroup
        \setlength{\tabcolsep}{1pt} 
        \renewcommand{\arraystretch}{1.2}
        \small 
        \begin{tabular}[t]{w{c}{0.7cm} | w{c}{0.5cm} w{c}{0.5cm} w{c}{0.5cm} w{c}{0.5cm} | w{c}{0.5cm} w{c}{0.5cm}}
            \toprule
            \textbf{$j$} & $P$ & $P_\pi$ & $P_{\text{r}}$ & $\tau$ & $Q$ & $\tau^{\text{-}1}$ \\
            \midrule
            \textbf{0} & \cellcolor{red!15}0 & \cellcolor{red!15}8 & \cellcolor{red!15}4 & \cellcolor{red!15}3 & \cellcolor{orange!15}0 & \cellcolor{orange!15}3 \\
            \textbf{1} & \cellcolor{blue!15}2 & \cellcolor{blue!15}7 & \cellcolor{blue!15}4 & \cellcolor{blue!15}2 & \cellcolor{teal!15}1 & \cellcolor{teal!15}5 \\
            \textbf{2} & \cellcolor{green!15}3 & \cellcolor{green!15}10 & \cellcolor{green!15}5 & \cellcolor{green!15}4 & \cellcolor{blue!15}7 & \cellcolor{blue!15}1 \\
            \textbf{3} & \cellcolor{orange!15}5 & \cellcolor{orange!15}0 & \cellcolor{orange!15}0 & \cellcolor{orange!15}0 & \cellcolor{red!15}8 & \cellcolor{red!15}0 \\
            \textbf{4} & \cellcolor{violet!15}6 & \cellcolor{violet!15}12 & \cellcolor{violet!15}5 & \cellcolor{violet!15}5 & \cellcolor{green!15}10 & \cellcolor{green!15}2 \\
            \textbf{5} & \cellcolor{teal!15}9 & \cellcolor{teal!15}1 & \cellcolor{teal!15}0 & \cellcolor{teal!15}1 & \cellcolor{violet!15}12 & \cellcolor{violet!15}4 \\
            \bottomrule
        \end{tabular}
        \endgroup
        \subcaption{Arrays ($P[6]=Q[6]=15$).}
    \end{minipage}
    
    \caption{Move structure and interval mapping. Colored arrows indicate the sorting permutation $\tau$, which maps input intervals $I[j]$ to their output intervals $\mathcal{O}[k]$.}
    \label{fig:move_structure}
\end{figure}

\begin{definition}
\label{def:move}
A \textit{move query} takes a position $i\in[0,n)$ and $j=P.\rank(i)$, and returns $\pi(i)$ and $j'=P.\rank(\pi(i))$.
\end{definition}

\subsubsection{Balancing}
\label{sec:ogbalance}

Bounding the time complexity of a move query depends on \textit{output intervals} and a \textit{balancing} procedure. Let $Q[0..r]$ be the sorted elements of $P_\pi$ appended with boundary element $Q[r]=n$, where $\tau: [0, r) \to [0, r)$ represents the sorting permutation such that $P_\pi[j] = Q[\tau(j)]$. Then $Q$ implicitly represents the set of \textit{output intervals} $\outputint=\{\outputint[0],\dots,\outputint[r-1]\}$ with $\min(\outputint[k]) < \min(\outputint[k+1])$ such that $\outputint[k]=[Q[k], Q[k+1])$ and $Q[k]=\min(\outputint[k])$. Intuitively, output intervals represent the exact images of input intervals under $\pi$; if $\tau(j)=k$, then $I[j]$ maps bijectively onto $\outputint[k]$ with $|I[j]|=|\outputint[k]|$. It follows that output intervals are disjoint and partition $[0,n)$.

Consider an arbitrary position $i\in I[j]$ mapping to $\pi(i)\in \outputint[k]$ such that $\tau(j)=k$. The cost of the linear scan to return $j' = P.\rank(\pi(i))$ after accessing $\PRANK[j]$ is $w = |\{ p \in P \mid Q[k] < p \leq \pi(i) \}|$. Informally, $w$ counts the number of input boundaries falling between the start of the output interval $\outputint[k]$ and the position $\pi(i)$. Without structural modification, a single output interval can be intersected by $O(r)$ input boundaries, causing a worst case query time of $\Theta(r)$~\cite{brown2022rlbwt}. Nishimoto and Tabei~\cite{nishimoto2021optimal} explained how to obtain a \textit{balanced} move structure satisfying $O(1)$-query time in $O(r)$-space by computing a new $P'\supseteq P$ which bounds any $w$ to a constant. Brown, Gagie, and Rossi~\cite{brown2022rlbwt} parameterized their result for a fixed integer $\alpha$, providing a time/space trade-off. Theorem \ref{thm:nt} summarizes these results, and Definition \ref{def:move_notation} details the notation for balanced move structures.

\begin{theorem}[Nishimoto and Tabei~\cite{nishimoto2021optimal}; Brown, Gagie, and Rossi~\cite{brown2022rlbwt}]
\label{thm:nt}
Let $\pi$ be a permutation on $[0,n)$ consisting of $r$ input/output intervals respectively represented by $P$ and $Q$. Then for any integer $\alpha \geq 2$ we can construct $P'\supseteq P$ and $Q' \supseteq Q$ such that any two consecutive elements $q_s < q_e$ in $Q'$ satisfy $|(q_s, q_e) \cap P'| < 2\alpha$. Further, $|P'|=|Q'|\leq \frac{\alpha\cdot r}{\alpha - 1}$; at most $\frac{r}{\alpha-1}$ new interval boundaries are introduced with respect to $P$ and $Q$.
\end{theorem}

\begin{definition}\label{def:move_notation}
    For a permutation $\pi$ with $r$ input intervals, let $\MOVE{\pi}$ denote any balanced move structure that computes its respective move query (Definition~\ref{def:move}) in $O(1)$-time and $O(r)$-space. The query notation $\MOVE{\pi}.move(i,j)$ returns the pair $(\pi(i),j')$.
\end{definition}

The process of balancing involves iterating until $P'$ does not contain a \textit{heavy} output interval $\outputint[k]=[q_s, q_e)$ of \textit{weight} $w=|(q_s, q_e)\cap P'|\geq 2\alpha$. For a heavy output interval $\outputint[k]$, balancing selects the $(\alpha + 1)$-th element $x$ in $(q_s, q_e) \cap P'$ and splits the output interval into $[q_s,x)$ and $[x,q_e)$. For the corresponding input interval $I[j]=[p_s,p_e)$ where $j=\tau^{-1}(k)$, we make the equivalent input interval split at $x'=p_s+(x-q_s)$ into $[p_s,x')$ and $[x',p_e)$. However, each split may cause a heavy output interval; we must find $q_s'= Q'.\pred(x')$ and $q_e'=Q'.\successor(x')$ and check the weight $w=|(q_s',q_e')\cap P'|$. Figure~\ref{fig:balancing_split} shows an example.

\begin{figure}[htbp]
    \centering
    \includegraphics[width=\columnwidth]{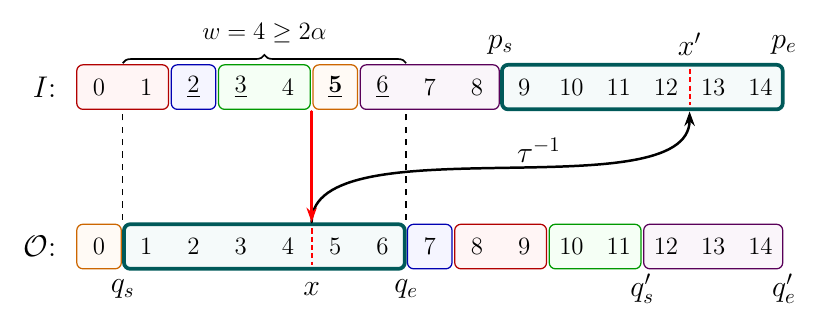}
    \caption{Balancing step for $\alpha=2$. The output interval with $(q_s,q_e)\cap P= \{2,3,5,6\}$ is heavy with weight $w=4$, triggering a split at $x=5$. The mapping $\tau^{-1}$ is used to find the corresponding input interval split point $x'$.}
    \label{fig:balancing_split}
\end{figure}

\subsubsection{Construction}
\label{sec:construction}

Constructing a balanced move structure relies on $P$, $P_\pi$, and $\tau$. Computing $P$ and $P_\pi$ given $\pi$ requires an $O(n)$-time preprocessing step ($O(n)$-space, or $O(r)$-space if $\pi$ is streamed), but they are often given as standard run-length encoded representations for move structures. From these, $\tau$ can be derived in $O(n)$-time and $O(r)$-space~\cite{nishimoto2021optimal}. However, for many permutation classes, $\tau$ (as well as $P$ and $P_\pi$) can be computed faster by exploiting structural properties that make $\pi$ decomposable into contiguously permuted intervals. To prevent an $O(n)$ derivation bottleneck, construction algorithms~\cite{nishimoto2021optimal,bertram2024move,brown2026boundingaveragestructurequery} assume $\tau$ is provided alongside the run-length encoded input $P$ and $P_\pi$, with the understanding that its derivation can be added as a preprocessing step to construction algorithms.

Given these inputs, Nishimoto and Tabei~\cite{nishimoto2021optimal} represent $P$ and $Q$ as balanced search trees. As established in Theorem~\ref{thm:nt}, their approach requires $O(r)$ balancing steps. Because each step performs insert, predecessor, and successor queries in $O(\log r)$-time, the total cost to balance the structure is $O(r \log r)$-time and $O(r)$-space (Theorem~\ref{thm:oldbalance}). Bertram, Fischer, and Nalbach~\cite{bertram2024move} later introduced an asymptotically equivalent but practically optimized algorithm using the balancing parameter $\alpha$. Their method processes $Q$ from left to right (smallest to largest), identifying heavy output intervals. If an insertion into $P$ creates a heavy interval among the already-balanced intervals, the algorithm uses recursive calls to restore balance.

\begin{theorem}[Nishimoto and Tabei~\cite{nishimoto2021optimal}]
\label{thm:oldbalance}
    Given $P$, $P_\pi$, and $\tau$ for a permutation $\pi$ with $r$ input intervals, a balanced move structure $\MOVE{\pi}$ can be constructed in $O(r \log r)$-time and $O(r)$-space.
\end{theorem}

\subsubsection{RLBWT Permutations}

Move structures were created for use with $\RLBWT$-based permutations $\LF$ and $\phi$ as well as their respective inverses $\FL$ and $\phi^{-1}$. As such, they can be represented in space proportional to the number of \BWT runs~\cite{nishimoto2021optimal}. Further, given the \RLBWT as input, these permutations permit efficient construction in $O(r)$-space. Existing results are summarized below.

\begin{lemma}
\label{lem:lf_input}
    Given $\RLBWT[0..r-1]$, then $P$, $P_{\pi}$, and $\tau$ with respect to $\LF$ or its inverse $\FL$ can be constructed in $O(r)$-time and $O(r)$-space~\cite{nishimoto2021optimal}.
\end{lemma}

\begin{theorem}\label{thm:fl_move_og}
Given $\RLBWT[0..r-1]$, then $\MOVE{_{LF}}$ or its inverse $\MOVE{_{FL}}$ can be constructed in $O(r \log r)$-time and $O(r)$-space.
\end{theorem}
\begin{proof}
    By~\Cref{thm:oldbalance,lem:lf_input}
\end{proof}

\begin{lemma}
\label{lem:phi_input}
    Given $\RLBWT[0..r-1]$ corresponding to $\BWT[0..n-1]$, then $P$, $P_{\pi}$ and $\tau$ with respect to $\phi$ or its inverse $\phi^{-1}$ can be constructed in $O(n)$-time and $O(r)$-space~\cite{brown2026boundingaveragestructurequery, sanaullah2026rlbwt}. 
\end{lemma}

\begin{theorem}\label{thm:phi_move_og}
Given $\RLBWT[0..r-1]$, then $\MOVE{_\phi}$ or its inverse $\MOVE{_{\phi^{-1}}}$ can be constructed in $O(n + r \log r)$-time and $O(r)$-space.
\end{theorem}
\begin{proof}
    By~\Cref{thm:oldbalance,lem:phi_input}
\end{proof}

\section{Optimal Move Structure Balancing}
\label{sec:balancing}

We present an optimal $O(r)$-time and space construction algorithm for move structures that satisfies the balancing criterion of Nishimoto and Tabei~\cite{nishimoto2021optimal}, improving upon their $O(r\log r)$-time baseline. Our algorithm relies on two key insights. First, we maintain linked lists instead of balanced search trees to support balancing operations. To support constant-time predecessor queries without search trees, we augment our lists with explicit cross-list predecessor pointers. Second, we balance $\pi$ and $\pi^{-1}$ simultaneously, whereas previous methods balanced only $\pi$.

Before detailing our formal structures, we highlight why replacing balanced search trees with linked lists is non-trivial. When an interval split forces an insertion in the opposite list via $\tau$ or $\tau^{-1}$, we must evaluate a predecessor query to identify if it has made some interval heavy. Without a search tree, resolving this predecessor in a standard linked list requires an $O(r)$-time scan. We avoid this by synchronizing the balancing of $\pi$ and $\pi^{-1}$ from left to right. By maintaining balance across both lists alongside predecessor pointers, any insertion landing in an already balanced region can resolve its predecessor and update other predecessor pointers upon a new insertion in $O(1)$-time via a local pointer walk.

\subsection{Linked Lists}\label{sec:linklist}

We maintain two linked lists with respect to a permutation $\pi$: $\PLIST$, representing the input intervals given by $P$, and $\QLIST$, representing the output intervals given by $Q$. Mappings between input and output intervals are described by $\tau$ and $\tau^{-1}$ pointers. Finally, we maintain cross-list predecessor pointers to resolve predecessor queries in constant time.~\Cref{fig:lists_diagram} visualizes the lists.

\begin{definition}\label{def:lists}
    The linked lists $\PLIST$ and $\QLIST$ are initialized to $r$ nodes with respect to $P$ and $Q$ ($|P|=|Q|=r+1$) such that the \xth{j} node $p\in\PLIST$ contains:
    \begin{itemize}[noitemsep, topsep=2pt, leftmargin=1.5em]
        \item $p.\text{next}$ pointing to the node corresponding to the \xth{(j+1)} input interval,
        \item $p.\text{idx} = P[j]$,
        \item $p.\tau$ pointing to the node corresponding to $\tau(j)$ in $\QLIST$,
        \item $p.\pred_Q$ pointing to the node corresponding to the output interval $k = Q.\rank(p.\text{idx})$ in $\QLIST$,
    \end{itemize}
    and the \xth{k} node $q \in \QLIST$ contains:
    \begin{itemize}[noitemsep, topsep=2pt, leftmargin=1.5em]
        \item $q.\text{next}$ pointing to the node corresponding to the \xth{(k+1)} output interval,
        \item $q.\text{idx} = Q[k]$,
        \item $q.\tau^{-1}$ pointing to the node corresponding to $\tau^{-1}(k)$ in $\PLIST$,
        \item $q.\pred_P$ pointing to the node corresponding to input interval $j = P.\rank(q.\text{idx})$ in $\PLIST$.
    \end{itemize}
    At the end of either list, $\text{next}.\text{idx}$ returns $n$, mirroring the function of the boundary element $P[r]=Q[r]=n$.
\end{definition}

\begin{figure*}[htbp]
    \centering
    \includegraphics[width=2\columnwidth]{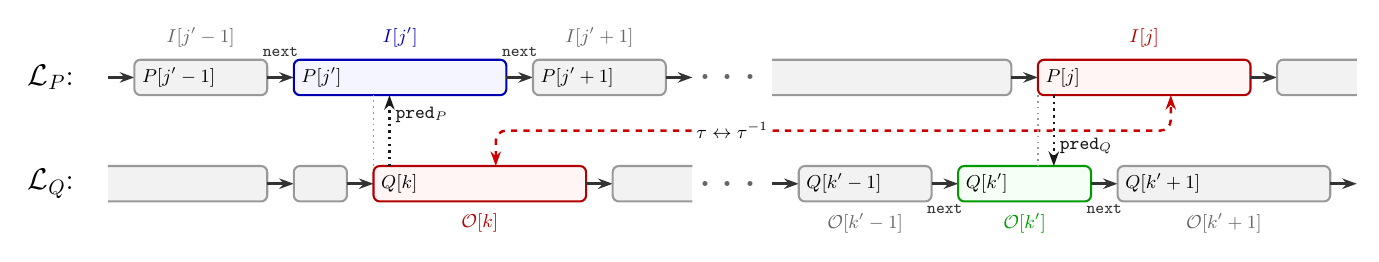}
    \caption{Visual representation of the linked lists $\PLIST$ and $\QLIST$, including next, $\tau$, and predecessor pointers as defined in \Cref{def:lists}.}
    \label{fig:lists_diagram}
\end{figure*}

First, we show that $\PLIST$ and $\QLIST$ can be initialized from $P, P_\pi,$ and $\tau$ in $O(r)$-time and space. Note that this is exactly the input specified in Theorem~\ref{thm:oldbalance}.

\begin{lemma}
\label{lem:init}
    Given $P$, $P_\pi$ and $\tau$, $\PLIST$ and $\QLIST$ can be initialized in $O(r)$-time and space.
\end{lemma}
\begin{proof}
    Construct $\tau^{-1}$ by setting $\tau^{-1}[\tau[i]]=i$ for $i\in[0,r)$. Then, compute $Q$ by setting $Q[k]=P_\pi[\tau^{-1}(k)]$ for all $k\in[0,r)$ and $Q[r]=P[r]=n$. We then initialize $\PLIST$ and $\QLIST$ by copying associated information from $P$ and $Q$, using $O(r)$-space random access arrays to set $\tau$ and $\tau^{-1}$ pointers. Finally, we compute the pointers $p.\text{pred}_Q$ and $q.\text{pred}_P$ using a two-finger walk, iterating through both lists simultaneously and selecting the interval with the minimum starting position at each step to derive each node's respective predecessor. Overall, these steps take $O(r)$-time and space
\end{proof}

\subsection{Balancing Algorithm}
\label{sec:algorithm}

We adapt the balance-on-the-fly approach of Bertram, Fischer, and Nalbach~\cite{bertram2024move}, scanning intervals left to right. We balance heavy intervals as we encounter them during this scan, or immediately when an insertion makes an interval heavy within an already-balanced region. For notational convenience, we assume $P=\{p.\text{idx} \mid p \in \PLIST\}\cup \{n\}$ and $Q=\{q.\text{idx} \mid q \in \QLIST\ \} \cup \{n\}$ stay up to date with insertions, though our algorithm does not maintain them explicitly. We maintain a \textit{balanced-up-to} parameter $t\in[0,n]$, which is initialized to $0$ and strictly increases after each iteration. By balancing $\PLIST$ and $\QLIST$ simultaneously from left to right, we enforce the invariant that no interval starting before $t$ remains heavy up to either $t$ or the interval endpoint:

\begin{invariant}\label[invariant]{inv:balance}
    Let $t$ be the balanced-up-to parameter. For every $p \in \PLIST$ with $p.\text{idx} < t$, we maintain $|\left(p.\text{idx},\min(t,p.\text{next}.\text{idx})\right)\cap Q|< 2 \alpha$. Symmetrically, for every $q \in \QLIST$ with $q.\text{idx} < t$, we maintain $|\left(q.\text{idx},\min(t,q.\text{next}.\text{idx})\right) \cap P| < 2 \alpha$.
\end{invariant}

We maintain pointers $p_t \in \PLIST$ and $q_t \in \QLIST$ to the input and output intervals in which the current position $t$ is contained. In each iteration, we select whichever of $p_t$ or $q_t$ has the smaller starting position, breaking ties by choosing the longer interval. If their lengths are identical, they are balanced by definition, and we advance both pointers. Otherwise, if the selected interval is heavy, we initiate a balancing step: we insert a new boundary into the current list, apply the corresponding mapping ($\tau$ or $\tau^{-1}$) to insert into the opposite list, and recurse if necessary as in Section~\ref{sec:ogbalance}. We detail the updates to $t$ in a following section, but note that $t$ strictly increases upon completing an iteration.

Synchronizing construction in both directions ensures that whenever an insertion via $\tau$ or $\tau^{-1}$ falls at a position $x' < t$, we can maintain cross-list predecessor pointers in $O(1)$-time. Conversely, when inserting at a position $x' \geq t$, we maintain all pointers except predecessors. However, our left-to-right scan naturally repairs these missing links as we reach them: because updating $p_t$ and $q_t$ requires traversing \texttt{next} pointers, we recover predecessor information on the fly, analogous to the initialization in Lemma~\ref{lem:init}. Thus, we maintain the following structural invariant:

\begin{invariant}\label[invariant]{inv:pred}
    Let $t$ be the balanced-up-to parameter. For every $p \in \PLIST$ with $p.\text{idx} \leq t$, the predecessor pointer $p.\pred_Q$ correctly references the predecessor of $p.\text{idx}$ in $\QLIST$. Symmetrically, for every $q \in \QLIST$ with $q.\text{idx} \leq t$, the pointer $q.\pred_P$ correctly references the predecessor of $q.\text{idx}$ in $\PLIST$.
\end{invariant}

We now show that our linked lists support the balancing steps of Section~\ref{sec:ogbalance} while preserving all $\PLIST$ and $\QLIST$ pointers alongside our invariants. Assume w.l.o.g.\ that in the current iteration we choose to balance $q_t \in \QLIST$ (the symmetric argument holds when selecting $p_t \in \PLIST$) and that the invariants hold up to $t$. First, we demonstrate how to execute the balancing operation and any required recursive steps in $O(1)$-time per insertion while preserving the invariants for the current $t$. Then, we show how to advance $t$, $q_t$, and $p_t$ to restore predecessor information and re-establish the invariants for the next iteration.

\subsubsection{Balancing Step}

A \textit{balancing step} is initiated when we encounter a heavy interval during our scan, inserting a new boundary into $\PLIST$ and $\QLIST$. This initial insertion may trigger further recursive insertions if it causes other intervals to become heavy at positions less than $t$. We show that this entire process preserves all structural invariants and requires only $O(\alpha)$-time per inserted interval.~\Cref{fig:balanced_up_to} illustrates how cross-list predecessor pointers are updated.

\begin{figure}[htbp]
    \centering
    \includegraphics[width=\columnwidth]{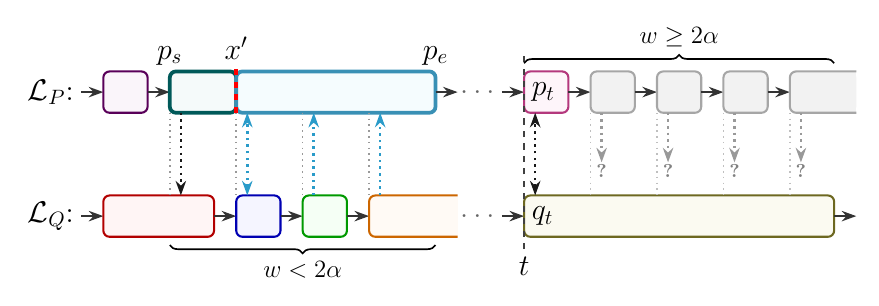} 
    \caption{List update for $\alpha=2$. Because the insertion at $x'$ falls within the already-balanced region ($x' < t$), we safely traverse valid predecessor pointers to repair cross-list routing for the newly split interval $[x', p_e)$ in $O(\alpha)$-time. Conversely, intervals beyond the balanced-up-to parameter $t$ may remain heavy with unreliably set predecessor pointers. The next iteration selects $q_t$ to balance.}
\label{fig:balanced_up_to}
\end{figure}

\begin{lemma}\label{lem:balancing_step}
    Suppose~\Cref{inv:balance,inv:pred} hold up to $t$. Then performing a balancing step on an interval $q_t \in \QLIST$ (or symmetrically, $p_t \in \PLIST$) preserves both invariants up to $t$ and completes in $O(\alpha)=O(1)$-time per inserted interval.
\end{lemma}

\begin{proof}
Assume w.l.o.g. that we balance $q_t$. Let $q_s=q_t.\text{idx}$ and $q_e=q_t.\text{next}.\text{idx}$. We compute the weight $|(q_s,q_e)\cap P|$ by accessing $p'=q_t.\pred_P$ and advancing $p'=p'.\text{next}$ until $p'.\text{idx} \geq q_e$ or the weight exceeds $2\alpha$; this takes $O(\alpha)$-time. If the output interval is heavy, we recover the $(\alpha + 1)$-th element $x$ in $(q_s, q_e) \cap P$ during this scan. Then, we insert node $q_x$ corresponding to $[x,q_e)$ into $\QLIST$ and update next pointers while setting $q_x.\text{idx}=x$. Let $p=q_t.\tau^{-1}$ with $p_s=p.\text{idx}$, $p_e=p.\text{next}.\text{idx}$, and $x'=p_s+(x-q_s)$. Symmetrically, we insert a node $p_{x'}$ with index $x'$ into $\PLIST$ for $[x',p_e)$. Finally, we link the two new nodes by setting $q_x.\tau^{-1}=p_{x'}$ and $p_{x'}.\tau=q_x$. All pointer and index updates take $O(1)$-time in total.

If $x' \geq t$, the balancing step terminates and Invariant~\ref{inv:balance} holds. If $x'<t$, we must repair cross-list predecessor pointers and verify whether further balancing is required. Because $p_s \leq x' < t$, Invariant~\ref{inv:pred} guarantees that the pointer $q'=p.\pred_Q$ can be used. By Invariant~\ref{inv:balance}, $|(p_s,\min(t, p_e)) \cap Q| < 2\alpha$. Thus, we can advance $q'=q'.\text{next}$ in $O(\alpha)$-time until we reach the node with the largest $q'.\text{idx} \leq x'$, after which we set $p_{x'}.\pred_Q=q'$. Continuing this scan through all remaining nodes in the span $x' \leq q'.\text{idx} < p_e$ with $q'.\text{idx} \leq t$, we update their predecessor pointers to $q'.\pred_P=p_{x'}$. Because the total number of nodes in this span is bounded by $2\alpha$, this pointer-repair process is $O(\alpha)$-time. Predecessor pointers beyond $t$ are restored in the update step (Lemma~\ref{lem:update_step}).

Finally, because inserting $x'$ into $\PLIST$ increases the weight of exactly one output interval, the interval $q^*= p_{x'}.\pred_Q$ spanning $x'$, we can check if $q^*$ became heavy in $O(\alpha)$-time. If heavy and $x' < t$, we recursively balance $q^*$. Because every recursive step inserts an interval to $\PLIST$ and $\QLIST$ while performing identical $O(\alpha)$-time operations as detailed above, the total work remains $O(\alpha)=O(1)$-time per inserted interval while maintaining~\Cref{inv:balance,inv:pred} up to $t$.
\end{proof}

\subsubsection{Update Step}\label{sec:update}

We now show how to advance $t$, $p_t$, and $q_t$ while preserving all invariants. 

\begin{lemma}\label{lem:update_step}
    Suppose Invariants~\ref{inv:balance} and \ref{inv:pred} hold up to $t$. After each balancing step iteration, we can advance $t$, $q_t$, and $p_t$ to preserve both invariants for the new value of $t$ in time proportional to the number of list nodes traversed.
\end{lemma}
\begin{proof}
Assume w.l.o.g. that we selected $q_t$ for balancing, and that it was heavy. Let $x$ be the position of the initial insertion into $\QLIST$ and $x'$ the position of the terminating insertion into $\PLIST$. Consider setting $t=x$. Let $q_s=q_t.\text{idx}$ and $q_e=q_t.\text{next}.\text{idx}$. If $q_s \leq x' < x$, then balancing terminated with the first split inside the newly created span $[q_s,x)$ since it initially has weight $\alpha$. Therefore, $[q_s,x)$ can have weight at most $\alpha + 1 < 2\alpha$ (since $\alpha\geq2$). Thus, it is not heavy, and balance is maintained up to $x$. If $x' \leq q_s$, balance is maintained by Lemma~\ref{lem:balancing_step}; if $x' > x$, balance up to $x$ is trivially maintained. In all cases, setting $t=x$ preserves Invariant~\ref{inv:balance} in $\QLIST$.

We now verify that Invariant~\ref{inv:balance} is also preserved in $\PLIST$ up to $t=x$. If $q_t.\text{idx}<p_t.\text{idx}$, any input interval $p\in\PLIST$ with $p.\text{idx} \in [q_s,x)$ must be balanced up to $x$: the weight of $p$ within this span is zero since it is wholly contained within $[q_s,x)$. If $q_t.\text{idx}=p_t.\text{idx}$ and $q_t.\text{next}.\text{idx}>p_t.\text{next}.\text{idx}$, identical reasoning holds. The case where both starting and ending positions match cannot occur, as $q_t$ would already be balanced by definition and not selected for a balancing step. Hence, $\PLIST$ also satisfies Invariant~\ref{inv:balance} up to $t=x$.

If $q_t$ was not heavy, we simply set $t=q_t.\text{next}.\text{idx}=q_e$. If $q_t.\text{idx}<p_t.\text{idx}$, any input interval $p\in\PLIST$ with $p.\text{idx} \in [q_s,q_e)$ must be balanced up to $q_e$: either it lies entirely within $q_t$ (yielding zero weight), or it is the unique interval intersecting $q_e$, which must have zero weight up to position $t=q_e$. The same reasoning applies when $q_t.\text{idx}=p_t.\text{idx}$ and $q_t.\text{next}.\text{idx} > p_t.\text{next}.\text{idx}$, while the case where both boundaries match is balanced by definition. Thus, Invariant~\ref{inv:balance} is preserved across all cases.

These cases also imply that $t$ strictly increases. If $q_t$ is not heavy, and we set $t=q_e$, then the new $p_t$ starts at $p_s\in[q_s,q_e]$. Thus, $|(p_s,t)\cap Q|=0$, so the $(\alpha + 1)$-th element of $(p_s,p_e)\cap Q$ occurs after $t$. After a heavy split, both current intervals start at $t=x$, so subsequent main-loop splits also advance $t$. Therefore, $t$ strictly increases, and neither $p_t$ nor $q_t$ moves backward.

Finally, regardless of whether a balancing step occurred, we advance $q_t$ and $p_t$ via their \texttt{next} pointers until we find the nodes whose intervals contain $t$. By using a two-finger walking approach that advances whichever pointer has the smaller starting position, cross-list predecessor pointers can be verified and updated concurrently during this scan. Once the new $q_t$ and $p_t$ are established, Invariant~\ref{inv:pred} is restored for the updated $t$. The time required to advance the pointers is strictly linear in the number of list nodes traversed.
\end{proof}

\subsection{Final Result}

The correctness of our algorithm follows directly by induction from our lemmas. Because the parameter $t$ initializes at $0$ and advances over interval starting positions to $n$, repeatedly applying Lemmas~\ref{lem:balancing_step} and \ref{lem:update_step} guarantees that Invariants~\ref{inv:balance} and \ref{inv:pred} hold upon termination; thus, the final lists represent balanced intervals with respect to $\pi$ and $\pi^{-1}$. We now evaluate the runtime and space requirements.

\begin{lemma}
\label{lem:runtime}
    Let $r'=|\PLIST|=|\QLIST|$ be the number of intervals after the algorithm finishes with $t=n$. Then the overall procedure completes in $O(r')$-time and $O(r')$-space for any balancing parameter $\alpha \geq 2$.
\end{lemma}
\begin{proof}
    Each balancing insertion takes $O(\alpha)$-time by Lemma~\ref{lem:balancing_step}. The total number of insertions, both original and recursive, across the entire execution is at most $O(r')$ since each introduces exactly one new interval to both $\PLIST$ and $\QLIST$. Furthermore, by Lemma~\ref{lem:update_step}, updating $p_t$ and $q_t$ requires time linear in the number of nodes advanced; because the updates of $p_t$ and $q_t$ after balancing only follow \texttt{next} pointers, this two-finger walk accesses at most $O(r')$ nodes in total across all iterations. Because cross-list predecessor pointers are updated concurrently during these scans, they add no asymptotic overhead. Finally, because space is strictly proportional to the number of nodes in $\PLIST$ and $\QLIST$, the entire procedure requires $O(r')$-time and space.
\end{proof}

It remains to show that the final number of intervals satisfies $r'\in O(r)$. Crucially, we observe that balancing an interval in $Q$ cannot create a heavy interval in $P$, and vice versa.

\begin{lemma}
    \label{lem:steps}
    When balancing $\pi$ and $\pi^{-1}$ with $r$ initial input intervals simultaneously via the algorithm of~\Cref{sec:algorithm} for any $\alpha \geq 2$, the procedure creates at most $\frac{2r}{\alpha-1}$ new intervals. Thus, upon completion, the final number of intervals is bounded by $r' \leq \frac{(\alpha+1)r}{\alpha-1}$.
\end{lemma}
\begin{proof}
    We adapt the potential function argument for move structure balancing attributed to Alhadi by Gagie~\cite{gagie2026merging} to simultaneous balancing. To bound the number of steps, define $f(k)=\sum_I\max(0,w_k(I)-\alpha)$, where $I$ ranges over the intervals in both lists after $k$ balancing steps, and $w_k(I)$ denotes the weight of $I$. Initially, $f(0) \leq 2r$. Splitting a heavy interval of weight $w\geq2\alpha$ results in new intervals with weights $\alpha$ and $w-\alpha-1$ respectively, decreasing the potential by at least $\alpha$. The corresponding insertion in the opposite list increases the weight of at most one interval in the selected direction by one, increasing the potential by at most one. 
    
    Further, splitting the corresponding interval in the opposite direction cannot increase its contribution to $f(k)$. Assume w.l.o.g. that we balance an arbitrary interval $q\in Q$; this operation never increases the weight of any interval in $P$. Balancing $q$ requires inserting a boundary $x$ into $Q$, and a corresponding boundary $x'$ into $P$. By definition, $x \in P$. If $p_e$ is the successor of $x$ in $P$, then $x \notin (x, p_e)$, so the weight of $[x,p_e)$ remains unchanged. Next, inserting $x'$ into $P$ splits some interval $[p_s, p_e)$ of weight $w = |(p_s, p_e) \cap Q|$ into two intervals whose weights sum to at most $w$. Thus, the combined contribution of the two parts to the potential is at most that of the original interval. By symmetry, the same holds when balancing $P$.

    Thus, $f(k+1)\leq f(k)-(\alpha-1)$ for every balancing step. Since $f(0)\leq2r$ and $f(k)\geq0$, at most $\frac{2r}{\alpha-1}$ balancing steps occur; this is twice the upper bound of balancing a single direction by Theorem~\ref{thm:nt}. Each step adds one interval to each list, giving $r'\leq r+\frac{2r}{\alpha-1}=\frac{(\alpha+1)r}{\alpha-1}$ intervals in each direction.
\end{proof}

Combining Lemma~\ref{lem:runtime} and Lemma~\ref{lem:steps}, our procedure executes in linear time and generates $O(r)$ total intervals, yielding our main algorithmic result:

\begin{theorem}
\label{thm:linear}
    Given $P$, $P_\pi$ and $\tau$ for a permutation $\pi$ with $r$ input intervals, then a balanced move structure $\MOVE{\pi}$ can be constructed in optimal $O(r)$-time and space.
\end{theorem}
\begin{proof}
    We initialize $\PLIST$ and $\QLIST$ per Definition~\ref{def:lists} in $O(r)$-time and space by Lemma~\ref{lem:init}. Our algorithm of~\Cref{sec:algorithm} balances $Q$ in $O(r)$-time and space by Lemma~\ref{lem:runtime} and Lemma~\ref{lem:steps} for any $\alpha \geq 2$. Specifically, starting from $t=0$, repeatedly applying Lemmas~\ref{lem:balancing_step} and \ref{lem:update_step} guarantees by induction that Invariants~\ref{inv:balance} and \ref{inv:pred} hold upon termination. We first perform an $O(r)$-time scan of $\PLIST$ to build an $O(r)$-space index operator $\rank_P(p)$, which returns the sequential integer position of any node $p \in \PLIST$. Then, by~\Cref{inv:balance}, we can extract $P'= \{p.\text{idx}~|~p \in \PLIST\} \cup \{n\}$, $P'_\pi=\{p.\tau.\text{idx}~|~p \in \PLIST\}$, and $\PRANK'=\{\rank_P(p.\tau.\pred_P)~|~ p \in \PLIST\}$ in $O(r)$-time and space. Let $Q'=\{q.\text{idx}~|~q \in \QLIST\}  \cup \{n\}$. By~\Cref{inv:balance}, any two consecutive elements $q_s < q_e$ in $Q'$ satisfy $|(q_s,q_e)\cap P'| < 2 \alpha$. Further, $|P'|, |P'_\pi|,|\PRANK'|\in O(r)$ by Lemma~\ref{lem:steps}. Hence, we return $\MOVE{\pi}$ in $O(r)$-time and space.
\end{proof}

\begin{corollary}
    Given $P$, $P_\pi$ and $\tau$ for a permutation $\pi$ with $r$ input intervals, then a balanced move structure $\MOVE{\pi^{-1}}$ can be constructed in optimal $O(r)$-time and space.
\end{corollary}
\begin{proof}
    We extract $Q'=\{q.\text{idx}~|~q \in \QLIST \} \cup \{n\}$, $Q'_{\pi^{-1}}=\{q.\tau^{-1}.\text{idx}~|~q \in \QLIST\}$, and $\QRANK'=\{\rank_Q(q.\tau^{-1}.\pred_Q)~|~q \in \QLIST\}$ analogous to the proof of Theorem~\ref{thm:linear}. Since the algorithm of~\Cref{sec:algorithm} balances in both directions, we satisfy that the move structure $\MOVE{\pi^{-1}}$ is balanced.
\end{proof}

\section{Optimal-time LCP Array from RLBWT}
\label{sec:optimal_lcp_rlbwt}

To demonstrate the algorithmic use of our optimal move structure construction, we apply it to the \LCP computation algorithm of Sanaullah et al.~\cite{sanaullah2026rlbwt}, which uses $O(n + r \log r)$-time and $O(r)$ working space to compute the \LCP array from the \RLBWT. Because the $O(r \log r)$ term is due to construction of a balanced move structure for $\FL$, applying the results of Section~\ref{sec:balancing} yields the first optimal $O(n)$-time algorithm to compute $\LCP[0..n-1]$ from $\RLBWT[0..r-1]$ within $O(r)$ working space. 

\begin{theorem}
\label{thm:lf_linear}
    Given $\RLBWT[0..r-1]$, a balanced move structure $\MOVE{LF}$ or its inverse $\MOVE{FL}$ can be constructed in optimal $O(r)$-time and space.
\end{theorem}
\begin{proof}
    By Lemma~\ref{lem:lf_input} and the optimal balancing procedure of Theorem~\ref{thm:linear}.
\end{proof}

\begin{theorem}
    Given $\RLBWT[0..r-1]$, a balanced move structure $\MOVE{\phi}$ or its inverse $\MOVE{\phi^{-1}}$ can be constructed in $O(n)$-time and $O(r)$-space.
\end{theorem}
\begin{proof}
    By Lemma~\ref{lem:phi_input} and the optimal balancing procedure of Theorem~\ref{thm:linear}.
\end{proof}

\begin{theorem}
\label{thm:irreducible_optimal}
    Given $\RLBWT[0..r-1]$ corresponding to $\BWT[0..n-1]$, the irreducible \PLCP array $\IPLCP[0..r-1]$ can be constructed in $O(n)$-time and $O(r)$-space.
\end{theorem}
\begin{proof}
    By substituting our optimal $O(r)$-time $\MOVE{FL}$ construction from Theorem~\ref{thm:lf_linear} for the $O(r \log r)$-time construction in the proof of Theorem~\ref{thm:irreducible_og} (see the Appendix).
\end{proof}

\begin{theorem}
    Given $\RLBWT[0..r-1]$ corresponding to $\BWT[0..n-1]$, $\LCP[0..n-1]$ can be constructed in optimal $O(n)$-time and $O(r)$ working space.
\end{theorem}
\begin{proof}
    By substituting our irreducible $\PLCP$ array computation from Theorem~\ref{thm:irreducible_optimal} into the procedure of Theorem~\ref{thm:lcp_og} (see the Appendix).
\end{proof}

We also obtain the same \LCP computation bound by adapting the case analysis of Kimura and I~\cite{kimura_et_al:LIPIcs.CPM.2026.10} to Sanaullah et al.'s~\cite{sanaullah2026rlbwt} approach; see the Appendix. Our balancing algorithm obtains this bound directly within their balanced move structure framework without a case distinction on $r$.

\section{Experiments}

\begin{table*}[t!]%
    \centering
    \begin{tabular}{l|rrrrrr}
        \hline
         & \multicolumn{6}{c}{\textbf{\# Chromosome-19 Seqs.}} \\
        & 32 & 64 & 128 & 256 & 512 & 1,000 \\
        \hline
        $\;n/10^6$  & 1,892.01 & 3,784.01 & 7,568.01 & 15,136.04 & 30,272.08 & 59,125.12 \\
        $\;r/10^4$  & 3,282.51 & 3,334.06 & 3,405.40 & 3,561.98 & 3,923.60 & 4,592.68 \\
        $\;n/r$     & 57.64 & 113.50 & 222.24 & 424.93 & 771.54 & 1,287.38 \\
        \hline
    \end{tabular}
    \bigskip
    \caption{Summarizes statistics of collections of human chromosome-19 sequences, where $n$ is the length of the \BWT and $r$ the number of \BWT runs.}
    \label{tab:dataset}
\end{table*}



To evaluate the efficiency of our $O(r)$-time and space move structure construction algorithm, we implemented the approach in {\CC} as part of the \texttt{Orbit} move structure library~\cite{brown2026boundingaveragestructurequery}, available at \url{https://github.com/drnatebrown/orbit}. 

\begin{figure*}[t!]
    \centering
    \includegraphics[width=\linewidth]{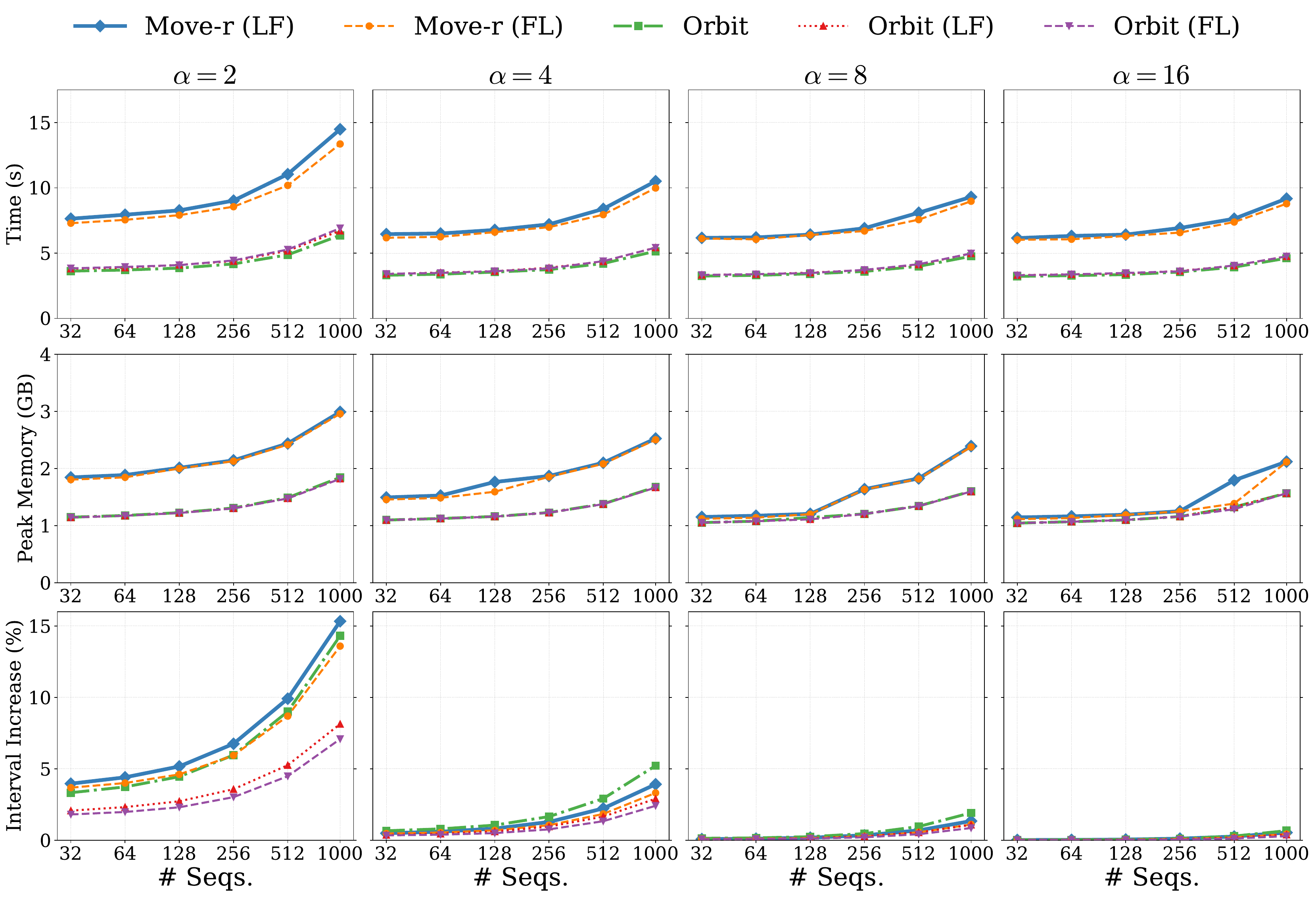}
    \caption{Scaling experiments for collections of chromosome-19 haplotypes for move structure balancing of \LF/\FL. \texttt{Orbit} denotes bidirectional balancing alone; direction-labeled \texttt{Orbit} variants additionally apply coarsening for that direction. Each column refers to a different choice of balancing parameter. Each row shows a different metric being measured, consisting of: runtime in seconds, peak memory in gigabytes, and the percentage interval increase compared to the number of initial intervals $r$. The $x$-axis describes the number of sequences in each collection. Time is averaged across $10$ runs.}
    \label{fig:chr19_lf_all}
\end{figure*}

\begin{figure*}[t!]
    \centering
    \includegraphics[width=\linewidth]{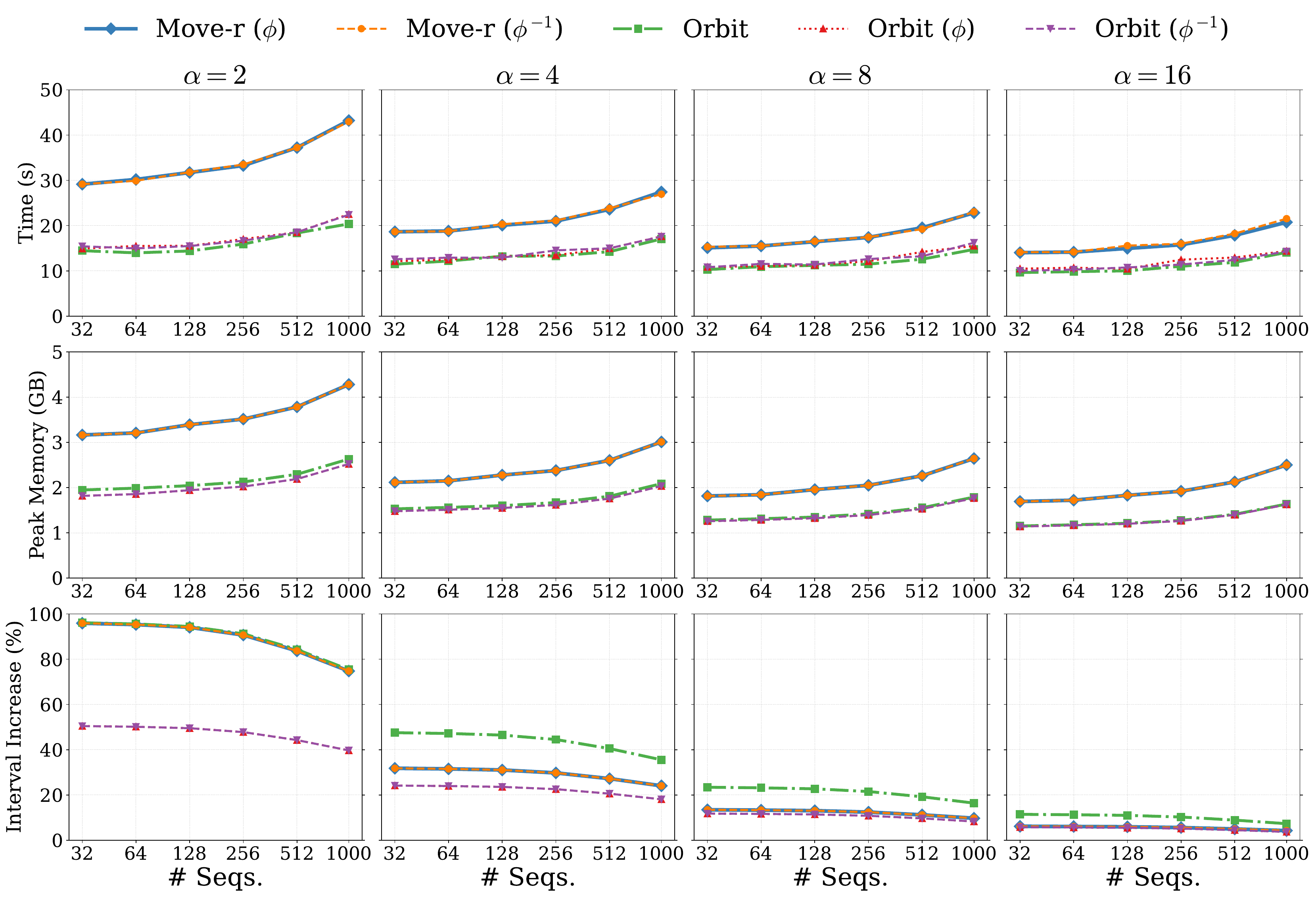}
    \caption{Scaling experiments for collections of chromosome-19 haplotypes for move structure balancing of $\phi$/$\phi^{-1}$. \texttt{Orbit} denotes bidirectional balancing alone; direction-labeled \texttt{Orbit} variants additionally apply coarsening for that direction. Each column refers to a different choice of balancing parameter. Each row shows a different metric being measured, consisting of: runtime in seconds, peak memory in gigabytes, and the percentage interval increase compared to the number of initial intervals $r$. The $x$-axis describes the number of sequences in each collection. Time is averaged across $10$ runs.}
    \label{fig:chr19_phi_all}
\end{figure*}

We simulate the linked lists of \Cref{sec:linklist} using bitpacked arrays. Beyond storage for the $r$ original intervals, we dynamically allocate space for insertions, appending each new node at the next free position. A preliminary scan estimates the number of insertions by counting the splits needed for each interval's initial weight, without performing insertions or updating other intervals' weights. We initially reserve $1.5$ times this estimate and grow the additional capacity by a factor of $1.5$ when full, capped at $2r/(\alpha-1)$ nodes per list by \Cref{lem:steps}. This bound allows us to replace pointers with array indices occupying $O(r\log r)$ bits in total. Including interval starting positions, the overall space is
$O(r\log n)$ bits.

Bidirectional balancing may introduce additional intervals when only one direction is needed. We therefore also evaluate \textit{coarsening}, a postprocessing scan that merges intervals separated by inserted boundaries. A merge is performed only if it preserves the target direction's $\alpha$-balancing constraints. In our results, \texttt{Orbit} refers to bidirectional balancing alone, while direction-labeled variants, such as \texttt{Orbit (LF)}, additionally apply coarsening for that direction. For these variants, reported runtimes and peak memory include both balancing and coarsening.

The $O(r \log r)$-time and $O(r)$-space balancing algorithm of Bertram, Fischer, and Nalbach~\cite{bertram2024move} was implemented in their \texttt{Move-r} tool, which we compare against. The tested \texttt{Move-r} implementation also splits long intervals,
which affects the intervals processed during balancing. We first compute the \RLBWT-based permutations \LF and $\phi$ and then pass them as input to the balancing algorithms. Since our approach balances in both directions in a single pass, we also run \texttt{Move-r} on the inverses \FL and $\phi^{-1}$. Crucially, we report each \texttt{Move-r} run independently as a standalone, single-direction runtime rather than summing them together; thus, our bidirectional construction is evaluated directly against single-direction baselines. We measure runtime, peak memory, and the percentage increase in the number of intervals on a server with two Intel Xeon Gold 6248R CPUs (3.00 GHz, 24 cores per socket, 48 cores total) and 1.5 TB of DDR4 memory, running Rocky Linux 8.8 (kernel 4.18.0-477.21.1.el8\_8). Programs were compiled with \texttt{g++} 11.3.0 using \texttt{-O3 -DNDEBUG -std=c++20}. We use one thread for all experiments, and evaluate using balancing parameter $\alpha \in \{2, 4, 8, 16\}$.

\begin{figure*}[t!]
    \centering
    \includegraphics[width=\linewidth]{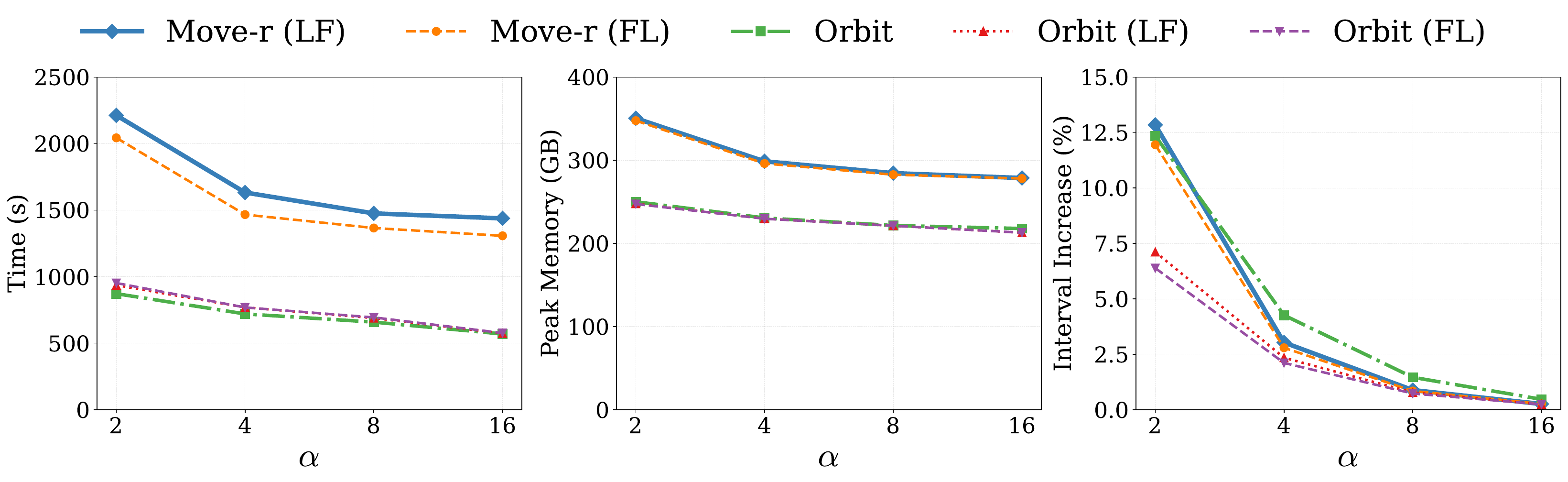}
    \caption{Move structure balancing of \LF/\FL on the HPRC release 2 collection. \texttt{Orbit} denotes bidirectional balancing alone; direction-labeled \texttt{Orbit} variants additionally apply coarsening for that direction. Each panel refers to a different metric being measured, consisting of: runtime in seconds, peak memory in gigabytes, and the percentage interval increase compared to the number of initial intervals $r$. The $x$-axis describes the balancing parameter used.} 
    \label{fig:hprc_results}
\end{figure*}

\subsection{Datasets}

To observe how the algorithms scale, we evaluate using the \RLBWT derived from a collection of human chromosome-19 haplotypes of up to $1,000$ sequences, using increasingly larger supersets. \Cref{tab:dataset} shows the size of each collection. We also evaluate the permutations \LF/\FL derived from the \RLBWT of the Human Pangenome Reference Consortium (HPRC)~\cite{liao2023draft} release 2 consisting of $466$ human haplotypes. This collection has size $n\approx 2.81$ trillion, with $r \approx 5.24$ billion and $n/r = 535.02$.

\subsection{Results}

For our chromosome-19 haplotype experiments, \Cref{fig:chr19_lf_all} shows the results for \LF/\FL, and \Cref{fig:chr19_phi_all} shows the results for $\phi$/$\phi^{-1}$. Crucially, across all experiments, \texttt{Orbit} is consistently faster than any single-direction run of \texttt{Move-r}, despite \texttt{Orbit} constructing both directions simultaneously in a single pass. This demonstrates an absolute wall-clock speedup over the search-tree baseline, even while performing the extra algorithmic work of bidirectional balancing. Across all tested configurations, \texttt{Orbit} uses less peak memory than \texttt{Move-r} for both \LF/\FL and $\phi$/$\phi^{-1}$, with or without coarsening.

Regarding structural size, bidirectional balancing has twice the worst-case insertion bound of single-direction balancing. However, for $\alpha=2$, \texttt{Orbit}'s interval increase is comparable to, and sometimes less than, that of single-directional \texttt{Move-r}; this may reflect more frequent splitting at small $\alpha$, with splits in one direction also helping to balance the other. For larger $\alpha$, the observed increase remains less than twice that of either single-direction build. Coarsening reduces the final interval counts below the corresponding \texttt{Move-r} baselines across all tested settings. Comparing permutations, \LF/\FL requires substantially fewer additional intervals than $\phi$/$\phi^{-1}$. This relative increase generally grows with collection size for \LF/\FL, but decreases for $\phi$/$\phi^{-1}$, both before and after coarsening.

\Cref{fig:hprc_results} shows consistent results on the much larger HPRC collection for \LF/\FL. Across all tested $\alpha$, \texttt{Orbit} is faster and uses less peak memory than either single-direction \texttt{Move-r} run, with or without coarsening. Despite having twice the worst-case insertion bound, bidirectional balancing produces a comparable interval increase to the single-direction baselines at $\alpha=2$. Coarsening reduces the final interval counts below the corresponding single-direction baselines.

\section{Conclusion}
We have presented the first optimal $O(r)$-time and space move structure construction algorithm. By replacing previous $O(r\log r)$-time search-tree approaches~\cite{nishimoto2021optimal,bertram2024move} with our bidirectional linked-list approach, we obtain the first optimal $O(n)$-time algorithm for streaming the $\LCP$ array from an $\RLBWT$ within $O(r)$ working space~\cite{sanaullah2026rlbwt}. Our experimental evaluation confirms that these theoretical gains translate directly into practice: our bidirectional construction is consistently faster than single-direction baselines while using less peak memory.

This improvement directly benefits compressed sequence analysis pipelines. For instance, in $\RLBWT$-based maximal match tools such as \texttt{mumemto}~\cite{shivakumar2025mumemto}, $\LCP$ enumeration remained the asymptotic bottleneck~\cite{brown2026boundingaveragestructurequery}; our results reduce this step to optimal-time in compressed space. Furthermore, because our algorithm balances both $\pi$ and $\pi^{-1}$ simultaneously, future work could collapse both directions into a unified \textit{invertible move structure}. This would allow constant-time queries in both directions without auxiliary satellite data or practical offset fields~\cite{brown2022rlbwt}. Such an invertible move structure requires balancing in both directions to retain the constant-time query guarantee. Both directions are useful, for example, when locating occurrences
with $\phi$ and $\phi^{-1}$ in \texttt{b-move}~\cite{depuydt2025b}.

Finally, our linear bound isolates the remaining challenge for $\RLBWT$-based construction. Because move structure construction is no longer an asymptotic bottleneck, the sole remaining barrier to constructing full $\LF, \FL, \phi,$ and $\phi^{-1}$ navigational structures from raw text in optimal $O(n)$-time and $O(r)$-space is improving upon the current-best $O(n + r \log r)$-time text to $\RLBWT$ construction algorithm in $O(r)$-space~\cite{nishimoto2022optimal}.


\appendix
\section{Appendix}

This Appendix serves to explain in greater detail the lemmas and theorems included in~\Cref{sec:lcp_rlbwt}, and is not required to understand the main result of the paper, optimal-time move structure construction. However, the proofs included in this section show exactly where our result can be applied to past methods, which we do so without full explanation in~\Cref{sec:optimal_lcp_rlbwt}.

To analyze improvements made with our $O(r)$-time balancing algorithm, we expand upon the specifics of Sanaullah et al.'s~\cite{sanaullah2026rlbwt} $O(n + r \log r)$-time and $O(r)$ working space algorithm to construct the LCP array from the \RLBWT. This includes re-framing their result to extend from past work and follow from provable lemmas/theorems to make obvious how our $O(r)$-time balancing result applies. Notably, recent work by Brown and Langmead~\cite{brown2026boundingaveragestructurequery} solved some existing problems present in the original work, which are included in our analysis. Important results from this section are restated in the main text (see \Cref{thm:fl_move_og,thm:irreducible_og,thm:lcp_og}). 

\subsection{LCP Array from BWT Background}

The first algorithm computing the \LCP array from \BWT or \SA in optimal $O(n)$-time was an $O(n)$-space approach by Kasai et al.~\cite{kasai2001linear} using $\SA[0..n-1]$, $\ISA[0..n-1]$, and $T[0..n-1]$. From $i=0 \to n-1$ it compares suffix $T[i..n-1]$ with its lexicographically preceding suffix (the suffix starting at $j=\SA[(\ISA[i] - 1) \mod n]$). The suffixes $T[i..n-1]$ and $T[j..n-1]$ are compared symbol by symbol until a mismatch occurs at $T[i+\ell + 1] \neq T[j+\ell+1]$ at which point we set $\LCP[\ISA[i]]=\ell$. The key insight is that if $\ell > 1$ symbols match at iteration $i$, then at least $\ell - 1$ symbols match in iteration $i+1$.

Consider that, if $\ell > 1$, then for $j=\SA[(\ISA[i] - 1) \mod n]$,  $T[i..n-1]=cwu$ and $T[j..n-1]=cwv$ have a shared prefix $cw$ of length $\ell$ where $c$ is a single symbol and $u \neq v$. Since $T[j..n-1]$ is the lexicographic predecessor of $T[i..n-1]$, it follows that $T[j+1..n-1]=wv$ is lexicographically smaller than $T[i+1..n-1]=wu$. By property of the suffix array, the lexicographically smallest suffix with the largest shared prefix for any suffix is its predecessor. Therefore, $T[i+1..n-1]$ shares a matching prefix of length at least $\ell -1$, the string $w$, with its lexicographic predecessor since there exists a lexicographically smaller suffix matching this length. Thus, we can safely skip $\ell - 1$ comparisons in iteration $i+1$ when $\LCP[\ISA[i]]=\ell$. A consequence of this insight is that comparisons of equal symbols  occur at most $n$ times (since each one increments $\ell$ and $i+\ell$ is non-decreasing from $0\to n-1$) and character comparisons of non-equal symbols occur at most $n$ times (since each one increments $i$). Therefore at most $2n$ symbol comparisons occur, implying an overall $O(n)$-time algorithm. 

Notice that the access pattern to $\LCP$ is not sequential, and further that $j=\SA[(\ISA[i] - 1) \mod n]$ can be replaced with $j=\phi(i)$. The $\phi$ algorithm by K\"arkk\"ainen, Manzini, and Puglisi~\cite{karkkainen2009permuted} is a variant of the algorithm of Kasai et al.~\cite{kasai2001linear} which constructs the $\PLCP$ array first using $\phi$. This allows them to omit $\ISA$ by taking as input $T[0..n-1]$ and $\SA[0..n-1]$ and constructing $\phi[0..n-1]$. Using $j=\phi(i)$ they make the same comparisons as prior, but now set $\PLCP[i]=\ell$ sequentially as $i$ is incremented. The final $\LCP$ array is then computed using the fact that $\LCP[i]=\PLCP[\SA[i]]$.

The above algorithms all rely on random access to the text $T$; when stored in memory to achieve constant time, this requires $n \log \sigma$ bits. Beller et al.~\cite{beller2013computing} introduced the first method to compute \LCP from \BWT or \SA which attempts to use smaller working space, achieving $O(n)$ bits of working space and $O(n \log \sigma)$-time. They avoid random access to $T$ by working over a \BWT and its wavelet tree. Prezza and Rosone~\cite{prezza2019space} achieved $o(n)$ bits of working space by combining the approach of Beller et al.~\cite{beller2013computing} with Belazzougui et al.'s~\cite{belazzougui2014linear} suffix tree interval enumeration algorithm. Although this does not require random access to $T$, it does require wavelet tree queries and the extra $O(\log \sigma)$-time factor. Further, although compact, their space bound is not in terms of $r$ which may be desired for large repetitive texts.

\subsection{PLCP Move Structures}

A known property of the \PLCP array is that if $\PLCP[i]$ is reducible, then $\PLCP[i]=\PLCP[i-1]-1$~\cite{karkkainen2009permuted}. 

\begin{corollary}
\label{cor:irreducible}
     Let $I = \{\, i \in [0,n) \mid \PLCP[i] \text{ is irreducible}\,\}$. Then where $j=I.\rank(i)$, $\PLCP[i]=\IPLCP[j] - (i - I[j])$~\cite{karkkainen2009permuted,gagie2020fully}.
\end{corollary}

A particularly useful relationship exists between $\phi$ intervals and irreducible $\PLCP$ values~\cite{gagie2020fully}. The \xth{j} input interval of $\phi$, $[P[j], P[j+1])$, is such that $\PLCP[P[j]]$ is irreducible and corresponds to $\IPLCP[j]$. Otherwise, for any $i\in(P[j], P[j+1])$, $\PLCP[i]$ is reducible. Thus, given \IPLCP alongside a $\phi$/$\phi^{-1}$ move structure, we can output $\PLCP[i]$ with respect to our current position $i$ in the permutation in $O(1)$-time~\cite{tatarnikov2022moni, sanaullah2025efficient}.

\subsection{Modified \texorpdfstring{$\phi$}{\textit{ϕ}} Algorithm}\label{sec:modified_phi}

A key difference from the $\phi$ algorithm of K\"arkk\"ainen, Manzini, and Puglisi~\cite{karkkainen2009permuted} is that Sanaullah et al.~\cite{sanaullah2026rlbwt} only need to compute the irreducible \PLCP array, $\IPLCP[0..r-1]$. Thus, they only compute $\phi(i)$ when position $i$ corresponds to an irreducible $\PLCP[i]$. Notice that these positions are the same as the input interval starts $P$ with respect to the permutation $\phi$, and thus we can compute $I = \{\, i \in [0,n) \mid \PLCP[i] \text{ is irreducible}\,\}$ in $O(n)$-time and $O(r)$-space. As such, the modified $\phi$ algorithm iterates over $k=0\to r-1$ setting $i=I[k]$ at the start of each iteration. The amount of work is equivalent to the $\phi$ algorithm, since this is equal to iterating over $i=0\to n-1$ and filling only irreducible positions. To update $\ell > 0$ for a new iteration $k$, we set $\ell=\ell - (i-I[k-1])$ to account for skipped reducible \PLCP positions.

\begin{lemma}
\label{lem:irr_pos}
    Given $\RLBWT[0..r-1]$ corresponding to $\BWT[0..n-1]$, we can construct $I = \{\, i \in [0,n) \mid \PLCP[i] \text{ is irreducible}\,\}$ in $O(n)$-time and $O(r)$-space by Lemma~\ref{lem:phi_input}.
\end{lemma}

At iteration $k$ with some updated $\ell$ and $i$, we must compute $j=\phi(i)$ to compare the suffixes from $T[i+\ell..n-1]$ and $T[j+\ell..n-1]$. Notice that $P_\pi$ with respect to $\phi$ gives the permutation mapping for any irreducible \PLCP position in $I$ and thus we can construct $\phi^+=\{\phi(i)~|~i \in I \}$ in $O(n)$-time and $O(r)$-space. Then, we compute $j=\phi(i)=\phi^+[k]$ with $\phi^+$. The algorithm proceeds by performing the same symbol comparisons as the $\phi$ algorithm, incrementing $\ell$ with each match. When a mismatch occurs, we set $\IPLCP[k]=\ell$.

\begin{lemma}\label{lem:phi_heads}
    Given $\RLBWT[0..r-1]$ corresponding to $\BWT[0..n-1]$, we can construct $\phi^+=\{\phi(i)~|~i \in I \}$ in $O(n)$-time and $O(r)$-space by \Cref{lem:phi_input,lem:irr_pos}.
\end{lemma}

\subsubsection{Simulating Random Access to Text}

However, to replace random access to $T$, Sanaullah et al.~\cite{sanaullah2026rlbwt} construct a move structure for $\FL$, $\MOVE{FL}$, in $O(r\log r)$-time and $O(r)$-space. Let us also store the $\BWT$ character of $\LF$ intervals permuted to their corresponding $\FL$ intervals, such that $\RLF=\{\BWT[\LF(i)]~|~i \in P'\}$ where $P'$ is the set of balanced input intervals with respect to $\MOVE{\FL}$. Informally, these are the characters of the \RLBWT (and characters at split points from balancing) as ordered by their permutation into the first column of the Burrows-Wheeler matrix (BWM). A character of the text $T[x]=\RLF[y]$ if $y=P'.\rank(\ISA[x])$ and thus we can extract characters of the text while using $\MOVE{FL}$ move queries through access with positional ranks.

\begin{theorem}[\Cref{thm:fl_move_og} of main text]
\label{lem:fl_move_og_app}
    Given $\RLBWT[0..r-1]$, then $\MOVE{FL}$ can be constructed in $O(r \log r)$-time and $O(r)$-space by Lemma~\ref{lem:lf_input} and Theorem~\ref{thm:oldbalance}.
\end{theorem}

\begin{lemma}\label{lem:rlf}
    Given $\RLBWT[0..r-1]$ and $\MOVE{FL}$, then $\RLF=\{\BWT[\LF(i)]~|~i \in P'\}$ can be constructed in $O(r)$-time and $O(r)$-space~\cite{nishimoto2021optimal,sanaullah2026rlbwt}.
\end{lemma}

We do not have access to $\ISA$ nor rank queries, so we need another method to find the correct position and interval in $\MOVE{FL}$ for text comparison. Let $s_{i}$ be initialized to $s_{i}=\ISA[0]=\FL(0)$ and $t_{i}$ be its rank in $\MOVE{FL}$. Then we can use these values to maintain $s_i=\ISA[i+\ell]$ while accessing $T[i+\ell]=\RLF[t_i]$ by setting $(s_i,t_i)=\MOVE{FL}.move(s_i,t_i)$ each time we extend $\ell$ forward, which we do at most $n$ times. Since our move structure is balanced, overall this takes $O(n)$-time. However, the value $j$ changes at each iteration and thus we cannot maintain coordinates in $\MOVE{FL}$ for $T[j+\ell]$ simply using consecutive move queries; indeed, this is where length capping cannot be applied to the algorithm~\cite{brown2026boundingaveragestructurequery}.

To support random access into the text when $j$ is updated, consider sampling every $\frac{n}{r}$ $\ISA$ values such that 
$$
\ISA_{n/r}
= \{\ISA[k \cdot \lceil n/r \rceil]~|~k \in [0,r)\}.
$$
Where $P'$ is with respect to $\MOVE{FL}$, let $\ISA_{\FL}=\{(j, P'.\rank(j))~|~j \in \ISA_{n/r}\}$ be such that $\ISA_{\FL}[k]$ stores the pair $\ISA_{n/r}[k]$ and the rank of its predecessor in $\MOVE{FL}$. Then if we need to access $T[j+\ell]$ we can lookup its predecessor sample location $k'=(j+\ell)/\ceil{\frac{n}{r}}$ using integer division, lookup the coordinates $(s_j,t_j)=\ISA_{FL}[k']$ in $\MOVE{FL}$ and then compute $\Delta=\left((j+\ell) \mod \ceil{\frac{n}{r}}\right)$ move queries to find the correct $(s_j',t_j')$ corresponding to $s_j'=\ISA[j+\ell]$. This requires at most $\ceil{\frac{n}{r}}$ steps due to the sampling scheme, and further, since we only require random access $r$ times when we update $j$ at irreducible \PLCP positions, we require at most $\ceil{\frac{n}{r}}\cdot r \in O(n)$ steps to simulate random access across the algorithm. Again, thanks to move structure balancing, overall this is $O(n)$-time. We can build $\ISA_{FL}$ itself by using $\MOVE{FL}$ to enumerate the suffix array in text order and sample in $O(n)$-time~\cite{sanaullah2026rlbwt}.

\begin{lemma}\label{lem:isa_samples}
    Given $\MOVE{FL}$ with corresponding balanced input intervals $P'$, then we can compute $\ISA_{n/r} = \{\ISA[i \cdot \lceil n/r \rceil]~|~i \in [0,r)\}$ and 
    $\ISA_{\FL}=\{(j, P'.\rank(j))~|~j \in \ISA_{n/r}\}$ in $O(n)$-time and $O(r)$-space~\cite{sanaullah2026rlbwt}.
\end{lemma}

\subsubsection{Algorithm Complexity}

Once $\IPLCP$ is constructed, one can use a move structure for $\phi$ (or, alternatively, $\phi^{-1})$ along with Corollary~\ref{cor:irreducible} to fill the values in $\LCP$ order from $i=n-1\to 0$~\cite{sanaullah2026rlbwt}. Notice that for the $\FL$ move structure, we required worst case guarantees provided from balancing. The $O(r\log\ r)$ factor is inherited from balancing in Theorem~\ref{thm:fl_move_og} with~\Cref{thm:oldbalance}.

\begin{theorem}[\Cref{thm:irreducible_og} of main text, Sanaullah et al.~\cite{sanaullah2026rlbwt}]\label{thm:irreducible_og_app}
    Given $\RLBWT[0..r-1]$ corresponding to $\BWT[0..n-1]$, we can compute the irreducible \PLCP array $\IPLCP[0..r-1]$ in $O(n + r \log r)$-time and $O(r)$-space.
\end{theorem}
\begin{proof}
    We construct $I$ and $\phi^+$ in $O(n)$-time and $O(r)$-space by \Cref{lem:irr_pos,lem:phi_heads}. We then construct $\MOVE{FL}$ in $O(r \log r)$-time and $O(r)$-space per~\Cref{thm:fl_move_og}. Finally, we construct in $O(r)$-space: $\RLF$ in $O(r)$-time using~\Cref{lem:rlf}, and $\ISA_{\FL}$ in $O(n)$-time using~\Cref{lem:isa_samples}. These are used to perform the modification of K\"arkk\"ainen, Manzini, and Puglisi's~\cite{karkkainen2009permuted} $\phi$ algorithm in $O(n)$-time. Overall, this results in $O(n+ r \log r)$-time and $O(r)$-space.
\end{proof}

\begin{theorem}[\Cref{thm:lcp_og} of main text, Sanaullah et al.~\cite{sanaullah2026rlbwt}]\label{thm:lcp_og_app}
    Given $\RLBWT[0..r-1]$ corresponding to $\BWT[0..n-1]$, we can construct $\LCP[0..n-1]$ in $O(n + r\log r)$-time and $O(r)$ working space.
\end{theorem}

\begin{proof}
    We apply the result of~\Cref{thm:irreducible_og} in $O(n + r \log r)$-time and $O(r)$-space. Then, we build a length capped move structure for $\phi^{-1}$ in $O(n)$-time and $O(r)$-space~\cite{brown2026boundingaveragestructurequery}. From~\Cref{cor:irreducible}, this allows us to enumerate from $i=n-1\to 0$ each value $\LCP[i]=\PLCP[\SA[i]]$ in $O(n)$ total time and $O(r)$ additional working space. Overall, this gives an $O(n + r \log r)$-time algorithm using $O(r)$ working space.
\end{proof}

\subsubsection{Alternative \LCP Array from RLBWT}

We also obtain optimal-time \LCP computation given an \RLBWT in $O(r)$ working space by adapting the case analysis of Kimura and I~\cite{kimura_et_al:LIPIcs.CPM.2026.10} to Sanaullah et al.'s~\cite{sanaullah2026rlbwt} approach. When $r \in o(n/\log n)$, the existing approach already achieves optimal time, since $r\log r \in O(n)$. When $r = \Omega(n / \log n)$, the balanced move structure for $\FL$ can be replaced with a predecessor-based representation using the data structure of Belazzougui and Navarro~\cite{belazzougui2015optimal}. This supports $O(\log \log_\word (n/r))=O(1)$-time $\FL$ queries, since $n/r=O(\log n)$ and $\word = \Theta(\log n)$, and requires $O(r)$ construction time and working space~\cite{gagie2020fully}. Thus, this substitution also yields $O(n)$-time \LCP computation with $O(r)$ working space from an \RLBWT.

\section*{Acknowledgments.}
We thank the reviewers for their insightful feedback, which has helped improve this work. In particular, we thank an anonymous reviewer for pointing out that linear-time balancing is sufficient but not necessary for the optimal LCP construction bound.

\bibliographystyle{siamplain}
\bibliography{refs}
\end{document}